\documentclass[runningheads]{SpringSTYLE15pg/llncs}
\let\proof\relax 
\let\endproof\relax 
\usepackage{graphicx}
\usepackage{microtype}
\usepackage{graphicx}
\usepackage{subfigure}
\usepackage{booktabs} 
\usepackage{mathtools,amsmath,amssymb,amsbsy}
\usepackage{times}  
\usepackage{helvet}  
\usepackage{courier}  
\usepackage{url}  
\usepackage{graphicx} 

\usepackage{natbib}
\usepackage{booktabs}       
\usepackage{amsfonts}       
\usepackage{nicefrac}       
\usepackage{microtype}      
\usepackage{enumerate}
\usepackage{microtype}
\usepackage[]{algorithm2e}
\usepackage{algorithmic}
\usepackage{graphicx}
\usepackage{subfigure}
\usepackage{booktabs} 
\usepackage{mathtools,amsmath,amssymb,amsbsy}
\usepackage{algorithmic}
\usepackage{amsthm}
\usepackage{color}
\usepackage{csquotes}
\usepackage{enumitem}
\usepackage{graphicx}
\usepackage{times}
\usepackage{helvet}
\usepackage{courier}
\usepackage{tcolorbox}
\usepackage{IEEEtrantools}

\usepackage[cal=boondoxo]{mathalfa}

\newtheorem*{remark-non}{Remark}

 \newtheorem{innercustomgeneric}{\customgenericname}
\providecommand{\customgenericname}{}
\newcommand{\newcustomtheorem}[2]{%
  \newenvironment{#1}[1]
  {%
   \renewcommand\customgenericname{#2}%
   \renewcommand\theinnercustomgeneric{##1}%
   \innercustomgeneric
  }
  {\endinnercustomgeneric}
}

\newcustomtheorem{customtheorem}{Theorem}
\newcustomtheorem{customlemma}{Lemma}
\newcustomtheorem{customproposition}{Proposition}
\newcustomtheorem{customdefinition}{Definition}
\usepackage{hyperref}

\begin{document}
\title{Cutting Your Losses: Learning Fault-Tolerant Control and Optimal Stopping under Adverse Risk}
%
\titlerunning{Cutting Your Losses}
%
\author{David Mguni\inst{1}}
\authorrunning{David Mguni}
%
\institute{PROWLER.io, Cambridge, UK. 
\email{davidmg@prowler.io}\\
}
\maketitle              

\begin{abstract}
Recently, there has been a surge in interest in safe and robust techniques within reinforcement learning (RL). 
Current notions of risk in RL fail to capture the potential for systemic failures such as abrupt stoppages from system failures or surpassing of safety thresholds and the appropriate responsive controls in such instances. We propose a novel approach to fault-tolerance within RL in which the controller learns a policy can cope with adversarial attacks and random stoppages that lead to failures of the system subcomponents. The results of the paper also cover fault-tolerant (FT) control so that the controller learns to avoid states that carry risk of system failures. By demonstrating that the class of problems is represented by a variant of stochastic games, we prove the existence of a solution which is a unique fixed point equilibrium of the game and characterise the optimal controller behaviour. We then introduce a value function approximation algorithm that converges to the solution through simulation in unknown environments. 
\end{abstract}

\section{Introduction}
Reinforcement learning (RL) provides the promise of adaptive agents being able to discover solutions merely through repeated interaction with their environment. RL has been deployed in a number of real-world settings in which, using RL, an adaptive agent learns to perform complex tasks, often in environments shared by human beings. Large scale factory industrial applications, traffic light control \citep{arel2010reinforcement}, robotics \citep{deisenroth2013survey} and autonomous vehicles \citep{shalev2016safe}  are notable examples of settings to which RL methods have been applied.

Numerous automated systems are however, susceptible to failures and unanticipated outcomes. Moreover, many real-world systems amenable to RL suffer the potential for random stoppages and abrupt failures; actuator faults, failing mechanical system components, sensor failures are few such examples. In these settings, executing preprogrammed behaviours or policies that have been trained in idealised simulated environments can prove vastly inadequate for the task of ensuring the safe execution of tasks. Consequently, in the presence of such occurrences, the deployment of RL agents introduces a risk of catastrophic outcomes whenever the agent is required to act so as to avoid adverse outcomes in unseen conditions. The important question of how to control the system in a way that is both robust against systemic faults and, minimises the risk of faults or damage therefore arises.

In response to the need to produce RL algorithms that execute tasks with safety guarantees, a significant amount of focus has recently been placed on safe execution, robust control and risk-minimisation \citep{garcia2015comprehensive}. Examples include $H_\infty$ control \citep{morimoto2001robust}, coherent risk, conditional value at risk \citep{tamar2015policy}. 
In general, these methods introduce an objective\footnote{With a Lagrangian approach, constraints are captured in the construction of the Lagrangian.} defined with an expectation measure that either penalises actions that lead to greater uncertainty or embeds a more pessimistic view of the world (for example, by biasing the transition predictions towards less desirable states). In both cases, the resulting policies act more cautiously over the horizon of the problem as compared to policies trained with a standard objective function.  

Despite the recent focus on safe methods within RL, the question of how to train an RL agent that can cope with random failures remains unaddressed. In particular, at present the question of how to produce an RL policy that can cope with an abrupt failure of some system subcomponent has received no systematic treatment. Similarly, the task of addressing how to produce RL policies that account for the risk of states in which such failures occur has not been addressed.


In this paper, we for the first time produce a method that learns optimal  policies in response to random and adversarial systems attacks that lead to stoppages of system (sub)components that may produce adverse events. Our method works by introducing an adversary that seeks to determine a stopping criterion to stop the system at states that lead to the worst possible (overall) outcomes for the controller. Using a game-theoretic construction, we then show how a policy that is robust against adversarial attacks that lead to abrupt failure can be learned by an adaptive agent using an RL updating method. In particular, the introduction of an adversary that performs attacks at states that lead to worst outcomes generates experiences for the adaptive RL agent to learn a \textit{best-response policy} against such scenarios. 

To tackle this problem, we construct a novel two-player stochastic game (SG) in which one of the players, the controller, is delegated the task of learning to modify the system dynamics through its actions that maximise its payoff and an adversary or `stopper' that enacts a strategy that stops the system in such a way that maximises the controller's costs. This produces a framework that finds optimal policies that are robust against stoppages at times that pose the greatest risk of catastrophe. 

The main contribution of the paper is to perform the first systematic treatment of the problem of robust control under worst-case failures. In particular, we perform a formal analysis of the game between the controller and the stopper.  Our main results are centered around a minimax proof that establishes the existence of a value of the game.  This is necessary for simulating the stopping action to induce fault-tolerance. Although minimax proofs are well-known in game theory \citep{shapley1953stochastic,maitra1970stochastic,filar1991nonlinear}, replacing a player's action set with stopping rules necessitates a minimax proof (which now relies on a construction of open sets) which markedly differs to the standard methods within game theory. Additionally, crucial to our analysis is the characterisation of the adversary optimal stopping rule (Theorem \ref{stopping_time_theorem}).

Our results tackle optimal stopping problems (OSPs) under worst-case transitions. OSPs are a subclass of optimal stochastic control (OSC) problems in which the goal is to determine a criterion for stopping at a time that maximises some state-dependent payoff \citep{peskir2006optimal}.

The framework is developed through a series of theoretical results: first, we establish the existence of a value of the game which characterises the payoff for the saddle point equilibrium (SPE). Second, we prove a contraction mapping property of a Bellman operator of the game and that the value is a unique fixed point of the operator. Third, we prove the existence and characterise the optimal stopping time. We then prove an equivalence between the game of control and stopping and worst-case OSPs and show that the fixed point solution of the game solves the OSP.

Finally, using an approximate dynamic programming method, we develop a simulation-based iterative scheme that computes the    optimal controls. The method applies in settings in which neither the system dynamics nor the reward function are known. Hence, the agent need only observe its realised rewards by interacting with the environment. 
\subsection{Related Work}
At present, the coverage of FT within RL is limited. In \citep{zhang2018reinforcement} RL is \textit{applied} to tackle systems in which faults might occur and subsequently incur a large cost. Similarly, RL is applied to a problem in \citep{yasuda2006homogeneous} in which an RL method for Bayesian discrimination which is used to segment the state and action spaces. Unlike these methods in which infrequent faults from the environment generate negative feedback, our method introduces an adversary that performs the task of simulating high-cost stoppages (hence, modelling faults) that induce an FT trained policy.

A relevant framework is a two-player optimal stopping game (Dynkin game) in which each player chooses one of two actions; to stop the game or continue \citep{dynkin1967game}. Dynkin games have generated a vast literature since the setting requires a markedly different analysis from standard SG theory. In the case with one stopper and one controller such as we are concerned with, the minimax proof requires a novel construction using open sets to cope with the stopping problem for the minimax result. Presently, the study of optimal control that combines control and stopping is limited to a few studies e.g. \cite{chancelier2002combined}. Similarly, games of control and stopping have been analysed in continuous-time \citep{bayraktar2011minimizing,baghery2013optimal,mguni2018viscosity}. In these analyses, all aspects of the environment are known and in general, solving these problems requires computing analytic solutions to non-linear partial differential equations which are often analytically insoluble and whose solutions can only be approximated numerically at very low dimensions.  

 Current iterative methods in OSPs (and approximated dynamic programming methods e.g. \cite{bertsekas2008approximate}) in unknown environments are restricted to risk-neutral settings \citep{tsitsiklis1999optimal} --- introducing a notion of risk (generated adversarially) adds considerable difficulty as it requires generalisation to an SG involving a controller and stopper which alters the proofs throughout. In particular, the solution concept is now an SG SPE, the existence of which must be established. As we show, our framework provides an iterative method of solving OSPs with worst-case transitions in unknown environments and hence, generalises existing OSP analyses to incorporate a notion of risk.   


\noindent\textbf{Organisation}

The paper is organised as follows: we firstly give a formal description of the FT RL problem we tackle and  the OSP with worst-case transitions and give a concrete example to illustrate an application of the problem. In Sec. \ref{section_SGs}, we introduce the underlying SG framework which we use within the main theoretical analysis which we perform in Sec. \ref{main_analysis}. Lastly, in Sec. \ref{approx_scheme}, we develop an approximate dynamic programming approach that enables the optimal controls to be computed through simulation, followed by some concluding remarks.

We now describe the main problem with which we are concerned that is, FT RL. We later prove an equivalence between
the OSPs under worst-case transitions and the FT RL problem and characterise the solution of each problem.

\subsection{Fault-Tolerant Reinforcement Learning} \label{sec:FT_RL_overview}
We concern ourselves with finding a policy that copes with abrupt system stoppages and failures at the worst possible states.  Unlike standard methods in RL and game theory that have fixed time horizons (or purely random exit times) in the following,  the process is stopped by a fictitious adversary that uses a stopping strategy or \textit{rule} to decide when to stop given its state observations. In order to generate an FT control, we simulate the adversary's action whilst the controller determines its optimal policy. This as we show, induces a form of control that is an FT best-response control.

A formal description is as follows: an agent exercises actions that influence the sequence of states visited by the system. At each state, the agent receives a reward which is dependent on the state and the chosen action. The agent's actions are selected by a policy $\pi:\mathcal{S}\times\mathcal{A}\to[0,1]$ --- a map from the set of states $\mathcal{S}$ and the set of actions $\mathcal{A}$ to a probability. We assume that the action set is a discrete compact set and that the agent's policy $\pi$ is drawn from a compact policy set $\Pi$. The horizon of the problem is $T\in\mathbb{N}\times\{\infty\}$. However, at any given point $\tau_S\leq T$ the system may stop (randomly) and the problem terminates where $\tau_S\sim f(\{0,\ldots,T\})$ is a measurable, random exit time and $f$ is some distribution on $\{0,\ldots,T\}$. If after $k\leq T$ time steps the system stops, the agent incurs a cost of $G(s_k)$ and the process terminates.

For any $s\in\mathcal{S}$ and for any $\pi\in\Pi$, the agent's \textbf{performance function} is given by:
\begin{align}\label{main_performance_function}
    J^{\tau_S,\pi}[s]=\mathbb{E}\left[\sum_{t=0}^{\tau_S\wedge T}\gamma^tR(s_t,a_t)+\gamma^{\tau_S\wedge T}G(s_{\tau_S\wedge T})\Bigg|s_0=s, a_t\sim \pi,\tau_S\sim f(\{0,\ldots,T\})\right],
\end{align}
where $a\wedge b:=\min\{a,b\}$, $\mathbb{E}$ is taken w.r.t. the transition function $P$. The performance function (\ref{main_performance_function}) consists of a \textbf{reward function} $R:\mathcal{S}\times\mathcal{A}\to\mathbb{R}$ which quantifies the agent's immediate reward when the system transitions from one state to the next, a \textbf{bequest function} $G:\mathcal{S}\to\mathbb{R}$ which quantifies the penalty incurred by the agent when the system is stopped and $\gamma\in [0,1[$, a discount factor. We assume $R$ and $G$ are bounded and measurable.

The FT control problem which we tackle is one in which the controller acts both with concern for abrupt system failures and stoppages. In particular, the analysis is performed in sympathy with addressing the problem of how the controller should act in two scenarios --- the first involves acting in environments that are susceptible to adversarial attacks or random stoppages in high costs states.  Such situations are often produced in various real-world scenarios such as engine failures in autonomous vehicles, network power failures and digital (communication) networks attacks. The second scenario involves a controller that seeks to avoid system states that yield a high likelihood of systemic (subcomponent) failure. Examples of this case include an agent that seeks to avoid performing tasks that increase the risk of some system failure, for example increasing stress that results in component failure or breakages within robotics.

To produce a control that is robust in these scenarios, it is firstly necessary to determine a stopping rule that stops the system at states that incur the highest overall costs. Applying this stopping rule to the system subsequently induces a response by the  controller that is robust against systemic faults at states in which stopping inflicts the greatest overall costs.  This necessitates a formalism that combines an OSP to determine an optimal (adversarial) stopping rule and secondly, a RL problem. Hence, problem we consider is the following:

$\quad$ Find $(\hat{k},\hat{\pi})\in\mathcal{V}   \times \Pi$ and $J^{\hat{k},\hat{\pi}}$ s.th.
\begin{equation}
    \underset{\pi\in\Pi}{\max}\left(\min_{k\in\mathcal{V}   }J^{k,\pi}[s]\right)=J^{\hat{k},\hat{\pi}}[s],\qquad \forall s \in \mathcal{S},\label{game_minimax_problem}
\end{equation}
where the minimisation is taken pointwise and $\mathcal{V}  $ is a set of stochastic processes of the form $v:\Omega\to\mathcal{T}$ where $\mathcal{T}\subseteq \{0,1,2\ldots\}$ is a set of stopping times. Note that by abusive of notation, we use $J^{{v},{\pi}}$ to mean $J^{{\tau},{\pi}}$ for any $\mathcal{T}\ni{\tau}\sim v$ where $v\in\mathcal{V}$. We also hereon employ the following shorthand $R({s,a})\equiv R_{s}^a$ for any $s\in\mathcal{S}$ and for any $a\in\mathcal{A}$. 

The dual objective (\ref{game_minimax_problem}) consists of finding both a stopping rule that minimises $J$ and an optimal policy that maximises $J$. By considering the tasks as being delegated to two individual \textit{players}, the problem becomes an SG between a controller that seeks to maximise $J$ by manipulating state visitations through its actions and an adversarial stopper that chooses a stopping rule to stop the process in order to minimise $J$. We later consider a setting in which neither player has up-front knowledge of the transition model or objective function but each only observes their realised rewards.

The results of this paper also tackle OSPs under a worst-case transitions --- problems in which the goal is to find a stopping rule $\hat{\tau}$ under the adverse \textit{non-linear expectation} $\mathcal{E}_P:=\underset{\pi\in\Pi}{\min}\;\mathbb{E}_{P,\pi}$ s.th.
    \begin{align}
    \hat{\tau}\in\underset{k\in\mathcal{V}   }{\arg\max}\mathcal{E}_P\left[\sum_{t=0}^{k\wedge T}\gamma^tR(s_t,a_t)+\gamma^{k\wedge T}G(s_{k\wedge T})\right].\label{robust_optimal_stopping_time}
\end{align}

Here, the agent seeks to find an optimal stopping time in a problem in which the system transitions according to an adversarial (worst-case) probability measure.
\subsection{Example: Control with random actuator failure}\label{section_examples}
To elucidate the ideas, we now provide a concrete practical example namely that of actuator failure within RL applications. 

Consider an adaptive learner, for example a robot that uses a set of actuators to perform actions. Given full operability of its set of actuators, the agent's actions are determined by a policy $\pi:S\times A\to [0,1] $ which maps from the state space $S$ and the set of actions $A$ to a probability. In many systems, there exists some risk of actuator failure at which point the agent thereafter can affect the state transitions by operating only a \textit{subset} of its actuators. In this instance, the agent's can only execute actions drawn from a subset of its action space $\hat{A}\subset A$ and hence, the agent is now restricted to policies of the form $\pi_{\rm partial}:{S}\times \hat{A}\to [0,1]$ --- thereafter its expected return is given by the value function $V^{\pi_{\rm partial}}$ (this plays the role of the bequest function $G$ in (\ref{main_performance_function})). In order to perform robustly against actuator failure, it is therefore necessary to consider a set of stopping times $\mathcal{T}\subseteq \{0,1,2,\ldots\}$ and a stopping criterion $\hat{\tau}:\Omega\to \mathcal{T}$ which determines the worst states for the agent's functionality to be impaired so that it can only use some subset of its set of actuators. 

The problem involves finding a pair  $(\hat{\tau},\hat{\pi})\in\mathcal{V}   \times\Pi$ --- a stopping time and policy s.th.
\begin{align}\nonumber
&\hspace{-3 mm}\underset{k'\in\mathcal{V}   }{\min}\left(\underset{\pi'\in\Pi}{\max}\;\mathbb{E}\left[H^{\pi',k'}(s)\right]\right)=\mathbb{E}\left[H^{\hat{\pi},\hat{\tau}}(s)\right];\qquad \forall s\in\mathcal{S},\label{example_objective}
\end{align}
where $s:=s_0, a_t\sim\pi'$ and $
H^{\pi,k}(s):=\sum_{t=0}^{k\wedge \infty}\gamma^tR(s_t,a_t)+\gamma^{k\wedge \infty}V^{\pi_{\rm partial}}(s_{k\wedge \infty})$. Hence the role of the adversary is to determine and execute the stopping action $\hat{\tau}$ that leads to the greatest reduction in the controller's overall payoff. The controller in turn learns to execute the policy $\hat{\pi}$ which involves playing a policy $\hat{\pi}_{\rm partial}\in \arg\max V^{{\pi}_{\rm partial}}$ after the adversary has executed its stopping action. The resulting policy $\hat{\pi}$ is hence robust against actuator failure at the worst possible states.

Embedded within problem (\ref{example_objective}) is an interdependence between the actions of the players --- that is, the solution to the problem is jointly determined by the actions of both players and their responses to each other. The appropriate framework to tackle this problem is therefore an SG \citep{shapley1953stochastic}.
 
\section{Discrete-Time Stochastic Games of control and stopping} \label{section_SGs}
In this setting, the state of the system is determined by a stochastic process $\{s_t|t=0,1,2,\ldots\}$ whose values are drawn from a state space $\mathcal{S}\subseteq \mathbb{R}^p$ for some $p\in\mathbb{N}$. 
The state space is defined on a probability space $(\Omega,\mathcal{B},P)$, where $\Omega$ is the sample space, $\mathcal{B}$ is the set of events and $P$ is a map from events to probabilities. We denote by $\mathcal{F}=(\mathcal{F}_n)_{n\geq 0}$ the filtration over $(\Omega,\mathcal{B},P)$ which  is an increasing family of $\sigma-$algebras generated by the random variables $s_1,s_2,\ldots$. We operate in a Hilbert space $\mathcal{V}$ of real-valued functions on $\mathbb{L}_2$,  i.e. a complete\footnote{A vector space is complete if it contains the limit points of all its Cauchy sequences.} vector space which we equip with a norm $\|\cdot\|:\mathcal{V}\to\mathbb{R}_{>0}\times\{0\}$ given by $\|f\|_{\mu}:=\sqrt{\mathbb{E}_{\mu}[f^2(s)}]$ and its inner product $\langle f,{f}^T\rangle_{\mu}:=\mathbb{E}_{\mu}\left[f(s){f}^T(s)\right]$ where $\mu:\mathcal{B}(\mathbb{R}^n)\to[0,1]$ is a probability measure. The problem occurs over a time interval $\{0,\ldots K\}$ where $K\in \mathbb{N}\times\{\infty\}$ is the time horizon. A \textbf{stopping time} is defined as a random variable $\tau:\Omega\to\{0,\ldots,K\}$ for which $\{\omega\in\Omega|\tau(\omega)\leq t\}\in \mathcal{F}_t$ for any $t\in \{0,\ldots,K\}$ --- this says that given the information generated by the state process, we can determine if the stopping criterion has occurred.  

An SG is an augmented Markov decision process which proceeds by two players tacking actions that \textit{jointly} manipulate the transitions of a system over $K$ rounds which may be infinite. At each round, the players receive some immediate reward or cost which is a function of the players' joint actions. The framework is zero-sum so that a reward for player 1 simultaneously represents a cost for player 2.

Formally, a two-player zero-sum SG is a $6-$tuple $\langle \mathcal{S},\mathcal{A}_{i\in\{1,2\}},P,R,\gamma\rangle$ where $\mathcal{S}=\{s_1,s_2,\ldots,\,s_n\}$ is a set of $n\in\mathbb{N}$ states, $\mathcal{A}_i$ is an action set for each player $i\in\{1,2\}$. The map $P:\mathcal{S}\times\mathcal{A}_1\times\mathcal{A}_2\times\mathcal{S}\to[0,1]$ is a Markov transition probability matrix i.e. $P(s';s,a_1,a_2)$ is the probability of the state $s'$ being the next state given the system is in state $s$ and actions $a_1\in\mathcal{A}_1$ and $a_2\in\mathcal{A}_2$ are applied by player 1 and player 2 (resp.). The function $R:\mathcal{S}\times\mathcal{A}_1\times\mathcal{A}_2$ is the one-step reward for player 1 and represents one-step cost for player 2 when player 1 takes action $a_1\in\mathcal{A}_1$ and player 2 takes action $a_2\in\mathcal{A}_2$ and $\gamma\in [0,1[$ is a discount factor. The goal of each player 1s to maximise its expected cumulative return --- since the game is antagonistic, the total expected reward received by player 1 which we denote by $J$, represents a total expected cost for player 2.

Denote by $\Pi_i$, the space of strategies for each player $i\in\{1,2\}$ . 
%
For standard SGs with Markovian transition dynamics, we can safely dispense with path dependencies in the space of strategies. In particular, it is well-known that for SGs, an equilibrium exists in Markov strategies even when the opponent can draw from non-Markovian strategies \citep{hill1979existence}.\footnote{There are some exceptions for games with payoff structures not considered here for example, limiting average (Ergodic) payoffs \citep{blackwell1968big}.} Consequently, we focus on the class of behavioural strategies that depend only on the current state and round, namely \textbf{Markov strategies}, hence for each player $i$, the strategy space $\Pi_i$ consists of strategies of the form $\pi_i:\mathcal{S}\times\mathcal{A}_i\to [0,1]$. 


In SGs, it is usual to consider the case $\mathcal{A}_1=\mathcal{A}_2$ so that the players' actions are drawn from the same set. We depart from this model and consider a game in which player 2 can choose a strategy which determines a time to stop the process contained within the set $\mathcal{T}\subseteq\{0,1,2,\ldots\}$ which consists of $\mathcal{F}-$ measurable stopping times. In this setting, player 1 can manipulate the system dynamics by taking actions drawn from $\mathcal{A}_1$ (we hereon use $\mathcal{A}$) and at each point, player 2 can decide to intervene to stop the game.

Let us define by ${\rm val}^+[J]:=\underset{k\in\mathcal{V}   }{\min}\hspace{0.2 mm}\underset{\pi\in\Pi}{\max}\;J^{k,\pi}$ the \textit{upper value function} and by ${\rm val}^-[J]:=\underset{\pi\in\Pi}{\max}\hspace{0.2 mm}\underset{k\in\mathcal{V}   }{\min}\;J^{k,\pi}$, the \textit{lower value function}. The upper (lower) value function represents the minimum payoff that player 1 (player 2) can guarantee itself irrespective of the actions of the opponent.  

The \textbf{value} of the game exists if we can commute the $\max$ and $\min$ operators:
\begin{align}
    {\rm val}^-[J]&=\underset{\pi\in\Pi}{\max}\min_{k\in\mathcal{V}   }J^{k,\pi}
    =\min_{k\in\mathcal{V}   }\underset{\pi\in\Pi}{\max}\;J^{k,\pi}={\rm val}^+[J].
\end{align}
We denote the value by $J^\star:={\rm val}^+[J]={\rm val}^-[J]$ and denote by $(\hat{k},\hat{\pi})\in\mathcal{V}   \times\Pi$ the pair that satisfies $J^{\hat{k},\hat{\pi}}\equiv J^\star$. The value, should it exist, is the minimum payoff each player can guarantee itself under the equilibrium strategy. In general, the functions ${\rm val}^+[J]$ and ${\rm val}^-[J]$ may not coincide. Should $J^\star$ exist, it constitutes an SPE of the game in which neither player can improve their payoff by playing some other control --- an analogous concept to a Nash equilibrium for the case of two-player zero-sum games. Thus the central task to establish an equilibrium involves unambiguously assigning a value to the game, that is proving the existence of $J^\star$.


\section{Main Analysis}\label{main_analysis}

In this section, we present the key results and perform the main analysis of the paper. Our first task is to prove the existence of a value of the game. This establishes a fixed or stable point which describes the equilibrium policies enacted by each player. Crucially, the equilibrium describes the maximum payoff that the controller can expect in an environment that is subject to adversarial attacks that stop the system or some subcomponent.  Unlike standard SGs with two controllers, introducing a stopping criterion requires an alternative analysis in which i) an equilibrium with Markov strategies in which one of the players uses a stopping criterion is determined and ii) the stopping criterion is characterised. It is well-known that introducing a stopping action to one of the players alters the analysis of SGs the standard methods of which cannot be directly applied (c.f. Dynkin games \citep{dynkin1967game}). 

Our second task is to perform an analysis that enables us to construct an approximate dynamic programming method. This enables the value function to be computed through simulation. This, as we show in Sec. \ref{approx_scheme}, underpins a simulation-based scheme that is suitable for settings in which the transition model and reward function is a priori unknown. Lastly, we construct an equivalence between robust OSPs and games of control and stopping. We defer some of the proofs to the appendix.

Our results develop the theory of risk within RL to cover instances in which the agent has concern the process at a catastrophic system state. Consequently, we develop the theory of SGs to cover games of control and stopping when neither player has up-front environment knowledge. We prove an equivalence between robust OSPs and games of control and stopping and demonstrate how each problem can be solved in unknown environments.






A central task is to prove that the Bellman operator for the game is a contraction mapping. Thereafter, we prove convergence to the unique value. 
Consider a Borel measurable function which is absolutely integrable w.r.t. the transition kernel $P^{\cdot}$ then $
    \mathbb{E}\left[J[s']|\mathcal{F}_t\right]=\int_\mathcal{S}J[s']P^a_{ss'}$, where $P^a_{ss'}\equiv P(s';s,a)$ is the probability of the state $s'$ being the next state given the action $a\in\mathcal{A}$ and the current state is $s$ .
In this paper, we denote by $
(PJ)(s):=\int_\mathcal{S}J[s']P^a_{sds'}$.

We now introduce the operator of the game which is of central importance:
\begin{align}
\hspace{-2 mm}TJ[s]:=\min\left\{\max_{a\in A} R^a_s+\gamma\sum_{s'\in\mathcal{S}}P^a_{ss'}J^{\tau,\pi}[s'],G(s)\right\},\qquad \forall s\in\mathcal{S}\label{bellman_op}
\end{align}

The operator $T$ enables the game to be broken down into a sequence of sub minimax problems. It will later play a crucial role in establishing a value iterative method for computing the value of the game.

\begin{proposition}\label{prop_bellman_contraction}
The operator $T$ in (\ref{bellman_op}) is a contraction.
\end{proposition}
\begin{proof}
We wish to prove that:
\begin{align}
\|TJ-T\bar{J}\|_{\pi}
\leq \gamma\|J-\bar{J}\|.
\end{align}
Firstly, we observe that:
\begin{align}\nonumber
&\begin{aligned}\Bigg\|\max_{a\in A}& \left\{R^a_s+\gamma\sum_{s'\in\mathcal{S}}P^a_{ss'}J^{\tau,\pi}[s'],G(s_k)\right\}-\left(\max_{a\in A}\left\{ R^a_s+\gamma\sum_{s'\in\mathcal{S}}P^a_{ss'}\bar{J}^\pi[s'],\bar{G}(s_k)\right\}\right)\Bigg\|
\end{aligned}
\\&\nonumber\leq \gamma\max_{a\in A}\left\|\sum_{s'\in\mathcal{S}}P^a_{ss'}\left(J^{\tau,\pi}_{s-1}[s']-\bar{J}^\pi_{s-1}[s']\right)\right\|
\leq \gamma\left\|J^{\tau,\pi}_{s-1}-\bar{J}^\pi_{s-1}\right\|,
\end{align}
using Cauchy-Schwartz (and that $\gamma\in [0,1[$) and (\ref{basic_max_result}). The result follows after applying Lemma \ref{max_min_inequality_2} and Lemma \ref{max_triple_inequality}.
\end{proof}

We now briefly  discuss strategies. A player strategy is a map from the opponent's policy set to the player's own policy set. In general, in two player games the player who performs an action first employs the use of a strategy. Typically, this allows the player to increase its rewards since their action is now a function of the other player's later decisions.  
Markov controls use only information about the current state and duration of the game rather than using information about the opponent's decisions or the game history.  Seemingly, limiting the analysis to Markov controls in the current game may restrict the abilities of the players to perform optimally.

Our first result however proves the existence of the value in Markov controls:
\begin{theorem}\label{existence_theorem}
\begin{align}
    {\rm val}^+[J]={\rm val}^-[J]\equiv J^\star.
\end{align}
\end{theorem}
Theorem \ref{existence_theorem} establishes the existence of the game which permits commuting the $\max$ and $\min$ operators of the objective (\ref{game_minimax_problem}). 
Crucially, the theorem secures the existence of an equilibrium pair $(\hat{\tau},\hat{\pi})\in\mathcal{V}   \times\Pi$, where $\hat{\pi}\in\Pi$ is the controller's optimal \textbf{Markov} policy when it faces adversarial attacks that stop the system. 
Additionally, Theorem \ref{existence_theorem} establishes the existence of a  given by $J^\star$, the computation of which, is the subject of the next section. 

\begin{proof}[Proof of Theorem 1]
We begin by noting the following inequality holds:
 \begin{align}
&{\rm val}^+[J]
=\underset{\tau\in\mathcal{T}}{\min}\;\underset{\pi\in\Pi}{\max}\;\mathbb{E}[J^{\tau,\pi}[s]]\geq \underset{\pi\in\Pi}{\max}\;\underset{\tau\in\mathcal{T}}{\min}\;\mathbb{E}[J^{\tau,\pi}[s]]
={\rm val}^-[J].\label{max_min_ineq_standard}
\end{align}
 The inequality follows by noticing $J^{k,\pi}\leq\underset{\pi\in\Pi}{\max}\;J^{k,\pi}$ and thereafter applying the $\min_{k\in\mathcal{T}}$ and ${\max}_{\pi\in\Pi}$ operators.
 
The proof can now be settled by reversing the inequality in (\ref{max_min_ineq_standard}). To begin, choose a sequence of open intervals $\{D_{m}\}_{m=1}^\infty$ s.th. for each $m=1,2,\ldots$ $\bar{D}_{m}$ is compact and $\bar{D}_{m}\supset \bar{D}_{m+1}$ and $[0,T]=\cap_{m=1}^\infty\bar{D}_{m}$ and define $\tau_D(m):=\inf_{k\in D_{m}}\mathbb{E}[J^{k,\pi}[s_0]]$.

We now observe that:
\begin{align}
&\;\mathbb{E}[J^{\tau,\hat{\pi}}[s]]    \begin{aligned}
\nonumber=\underset{\pi\in\Pi}{\max}\;\mathbb{E}&\left[\sum_{t=0}^{\tau_D(m)}\gamma^t(R(s_t,a_t)+G(s_{\tau_D(m)}))\right]
-\mathbb{E}\left[\sum_{t=\tau}^{\tau_D(m)}\gamma^t(R(s_t,a_t)+G(s_{\tau_D(m)}))\right]
\end{aligned}
\\&\geq\begin{aligned}
\nonumber\mathbb{E}&\left[J^{\tau_D(m),\pi}[s]\right]
-\left|\mathbb{E}\left[\sum_{t=\tau}^{\tau_D(m)}\gamma^t(R(s_t,a_t)+G(s_{\tau_D(m)}))\right]\right|
\end{aligned}
\\&\begin{aligned}\nonumber\geq\mathbb{E}&\left[J^{\tau_D(m),\pi}[s]\right]
-\sum_{t=\tau}^{\tau_D(m)}\gamma^t\left|\mathbb{E}[R(s_t,a_t)]+\mathbb{E}\left[G(s_{\tau_D(m)})\right]\right|\end{aligned}
\\&\nonumber\geq\mathbb{E}\left[J^{\tau_D(m),\pi}[s]\right]-\sum_{t=\tau}^{\tau_D(m)}\gamma^t\left(\mathbb{E}\left[\left|R(s_0,\cdot)\right|\right]+\mathbb{E}\left[\left|G(s_0)\right|\right]\right)
\\&\nonumber=\mathbb{E}\left[J^{\tau_D(m),\pi}[s]\right]+\gamma^{\tau_D(m)+1}\frac{1-\gamma^{\tau-\tau_D(m)}}{1-\gamma}c
\\&=\lim_{m\to\infty}\inf\mathbb{E}[J^{\tau_D(m),\pi}[s]]
+\lim_{m\to\infty}\left[\gamma^{\tau_D(m)+1}\frac{1-\gamma^{\tau-\tau_D(m)}}{1-\gamma}\right]c
\geq\mathbb{E}[J^{\tau,\pi}[s]],\nonumber
\end{align}
where we have used the stationarity property and, in the limit $m\to \infty$ and, in the last line we used the Fatou lemma. The constant $c$ is given by $c:=(\mathbb{E}[R(s_0,\cdot)]+\mathbb{E}[G(s_0)])\in \mathbb{L}$.

Hence, we now find that
\begin{align}
    \mathbb{E}[J^{\tau,\hat{\pi}}[s]]\geq\mathbb{E}[J^{\tau,\pi}[s]].\label{part_1_summary}
\end{align} 
Now since (\ref{part_1_summary}) holds $\forall \pi \in\Pi$ we find that: 
\begin{align}
    \mathbb{E}[J^{\tau,\hat{\pi}}[s]]\geq\underset{\pi\in\Pi}{\max}\;\mathbb{E}[J^{\tau,\pi}[s]].\label{pen_step_part_1}
\end{align} 
Lastly, applying $\min$ operator we observe that:
\begin{align}
    \mathbb{E}[J^{\hat{\tau},\hat{\pi}}[s]]\geq\underset{\tau\in\mathcal{T}}{\min}\;\underset{\pi\in\Pi}{\max}\;\mathbb{E}[J^{\tau,\pi}[s]]={\rm val}^+[J].
\label{val_plus_min_max_ineq}\end{align} 
It now remains to show the reverse inequality holds:
\begin{align}
    \mathbb{E}[J^{\hat{\tau},\hat{\pi}}[s]]\leq\underset{\pi\in\Pi}{\max}\;\underset{\tau\in\mathcal{T}}{\min}\;\mathbb{E}[J^{\tau,\pi}[s]]={\rm val}^-[J].
\end{align}
Indeed, we observe that
\begin{align}
&\mathbb{E}\left[J^{\hat{\tau},\hat{\pi}}[s]\right]\label{first_proof_payoff_bound_start}
\begin{aligned}
\leq \underset{\tau\in\mathcal{T}}{\min}\;&\mathbb{E}\left[J^{\tau\wedge m,\hat{\pi}}[s]\right]
+\mathbb{E}\left[\sum_{t=m}^\infty\gamma^t\left(|R(s_t,a_t)|+|G(s_t)|\right)\right]
\end{aligned}
\\&\leq \lim_{m\to\infty}\left[\underset{\tau\in\mathcal{T}}{\min}\;\mathbb{E}\left[J^{\tau\wedge m,\hat{\pi}}[s]\right]+c(m)\right]\label{first_proof_payoff_bound_mid}
\\&= \underset{\tau\in\mathcal{T}}{\min}\;\mathbb{E}\left[J^{\tau,\hat{\pi}}[s]\right]\leq \underset{\pi\in\Pi}{\max}\;\underset{\tau\in\mathcal{T}}{\min}\;\mathbb{E}\left[J^{\tau,{\pi}}[s]\right],\label{min_max_ineq_proof}
\end{align}
since $\gamma\in [0,1[$,where $c(m):=\frac{\gamma^{m}}{1-\gamma}(\mathbb{E}[|R(s_0,\cdot)|]+\mathbb{E}[|G(s_0)|])$ (using the stationarity of the state process) and where we have used Lebesgue's Dominated Convergence Theorem in the penultimate step.

Hence, by (\ref{min_max_ineq_proof}) we have that:
\begin{align}
\hspace{-3 mm}\mathbb{E}\left[J^{\hat{\tau},\hat{\pi}}[s]\right]\leq \underset{\pi\in\Pi}{\max}\;\underset{\tau\in\mathcal{T}}{\min}\;\mathbb{E}\left[J^{\tau,{\pi}}[s]\right]={\rm val}^-[J].    \label{val_minus_max_min_ineq}
\end{align}
Hence putting (\ref{val_plus_min_max_ineq}) and (\ref{val_minus_max_min_ineq}) together gives:
\begin{align}\nonumber
&\hspace{-3 mm}{\rm val}^-[J]=\underset{\pi\in\Pi}{\max}\;\underset{\tau\in\mathcal{T}}{\min}\;\mathbb{E}\left[J^{\tau,{\pi}}[s]\right]\\&\hspace{-3 mm}\geq    \mathbb{E}[J^{\hat{\tau},\hat{\pi}}[s]]\geq\underset{\tau\in\mathcal{T}}{\min}\;\underset{\pi\in\Pi}{\max}\;\mathbb{E}[J^{\tau,\pi}[s]]={\rm val}^+[J].\label{min_max_ineq_unstandard}
\end{align} 
After combining (\ref{min_max_ineq_unstandard}) with (\ref{max_min_ineq_standard}) we deduce the thesis.
\end{proof}
We can now establish the optimal strategies for each player. To this end, we now define best-response strategies which shall be useful for further characterising the equilibrium:
\begin{definition}\label{BR_strategies}The set of \textbf{best-response (BR) strategies} for player 1 \textit{against the stopping time} $\tau\in\mathcal{V}   $ (BR strategies for player 2 \textit{against the policy} $\pi\in\Pi$) is defined by: 
\begin{align}
\hat{\pi}\in\underset{\pi'\in\Pi}{\arg\hspace{-0.35 mm}\max}\;\mathbb{E}[J^{\tau,\pi'}[s]] \quad \text{(resp.,} \hat{\tau}\in\underset{\tau'\in\mathcal{V}   }{\arg\hspace{-0.35 mm}\min}\;\mathbb{E}[J^{\tau',\pi}[s]]), \qquad \forall s\in\mathcal{S}.    
\end{align}
\end{definition}    
The question of computing the value of the game remains. To this end, we now prove that repeatedly applying $T$ produces a sequence that converges to the value. 
In particular, the game has a \textit{fixed point property} which is stated in the following:
\begin{theorem}\label{fixed_point_theorem}
1. The sequence $(T^nJ)_{n=0}^\infty$ converges (in $\mathbb{L}_2$).\label{limit_statement_theorem}\\
2. There exists a unique function $J^\star\in \mathbb{L}_2$ s.th.
\begin{align}
    J^\star=TJ^\star\;\;\text{and}\;\; \underset{n\to\infty}{\lim}T^nJ=J^\star.
\end{align}
\end{theorem}
Theorem \ref{fixed_point_theorem} establishes the existence of a fixed point of $T$ and that the fixed point coincides with the value of the game. Crucially, it suggests that $J^\star$ can be computed by an iterative application of the Bellman operator which underpins a value iterative method. We study this aspect in Sec. \ref{approx_scheme} where we develop an iterative scheme for computing $J^\star$.
\begin{proof}[Proof of Theorem \ref{fixed_point_theorem}]
\underline{\textit{Part 1:}} We note that the contraction property of $T$ (c.f. Prop. \ref{prop_bellman_contraction}) allows us to demonstrate that the game has a unique fixed point to which a sequence $(T^nJ)_{n=0}^\infty$ converges (in $\mathbb{L}_2$). In particular, by Prop. 1 we have that $\|T^2J-TJ\|\leq\gamma\|TJ-J\|$ which proves that the sequence $(T^nJ)_{n=0}^\infty$ converges to a fixed point. 

\underline{\textit{Part 2:}} We observe that the fixed point is unique since if $\exists J,M\in \mathbb{L}_2$ s.th. $TJ=J$ and $TM=M$ we find that $\|M-J\|=\|TM-TJ\|=\gamma\|M-J\|$, so that $M=J$ (since $\gamma\in [0,1[$) which gives the desired result.

Adopting notions in dynamic programming, denote by:
\begin{align}
    &T^nJ[s]\nonumber
    =\underset{\tau\in\mathcal{T}}{\min}\underset{\pi_0,\pi_1,\ldots,\pi_{n-1}}{\max}\mathbb{E}\left[\sum_{t=0}^{\{n-1\wedge\tau\}}\gamma^tR(s_t,a_t)+\gamma^nJ(s_{n\wedge\tau})\right].
\end{align}
We begin the proof by invoking similar reasoning as (\ref{first_proof_payoff_bound_start}) - (\ref{first_proof_payoff_bound_mid}) to deduce that:
\begin{align}\nonumber
&\mathbb{E}\left[J^{\hat{\tau},\hat{\pi}}[s]\right]\leq \underset{\tau\in\mathcal{T}}{\min}\;\mathbb{E}\left[J^{\tau\wedge n,\hat{\pi}}[s]\right]+\frac{\gamma^{n}}{1-\gamma}c,
\end{align}
where $c:=(\mathbb{E}[|R(s_0,\cdot)|]+\mathbb{E}[|G(s_0)|])$.
Hence,
\begin{align}
&T^nJ[s]\leq \underset{\pi\in\Pi}{\max}\;\underset{\tau\in\mathcal{T}}{\min}\;\mathbb{E}\left[J^{\tau,\pi}[s]\right]+\frac{\gamma^{n}}{1-\gamma}c
= J^\star[s]+\frac{\gamma^{n}}{1-\gamma}c. \label{T_op_upper_bound}
\end{align}
By analogous reasoning we can deduce that:
\begin{align}
&T^nJ[s]\geq \underset{\tau\in\mathcal{T}}{\min}\;\underset{\pi\in\Pi}{\max}\;\mathbb{E}\left[J^{\tau,\pi}[s]\right]-\frac{\gamma^{n}}{1-\gamma}c
=J^\star[s]-\frac{\gamma^{n}}{1-\gamma}c. \label{T_op_lower_bound}
\end{align}
Putting (\ref{T_op_upper_bound}) and (\ref{T_op_lower_bound}) together implies:
\begin{align}
J^\star[s]-\frac{\gamma^{n}}{1-\gamma}c\leq T^nJ[s]\leq J^\star[s]+\frac{\gamma^{n}}{1-\gamma}c.\label{T_inequality}
\end{align}
By Lemma \ref{T_operator_properties}, i.e. invoking the monotonicity and constant shift properties of $T$, we can apply $T$ to (\ref{T_inequality}) and preserve the inequalities to give:
\begin{align}
TJ^\star[s]-\frac{\gamma^{n}}{1-\gamma}c\leq T^{n+1}J[s]\leq TJ^\star[s]+\frac{\gamma^{n}}{1-\gamma}c.\label{T_inequality_plus_1}
\end{align}
After taking the limit in (\ref{T_inequality_plus_1}) and, using the sandwich theorem of calculus, we deduce the result.
\end{proof}
\begin{definition}
The pair $(\hat{\tau},\hat{\pi})\in\mathcal{V}   \times\Pi$ is an SPE iff:
\begin{align}
    J^{\hat{\tau},\hat{\pi}}[s]= \underset{\pi\in\Pi}{\max}\;J^{\hat{\tau},\pi}[s]=\underset{\tau\in\mathcal{V}   }{\min}\;J^{\tau,\hat{\pi}}[s], \qquad \forall s\in\mathcal{S}.
\end{align}
\end{definition}
An SPE therefore defines a strategic configuration in which both players play their BR strategies. With reference to the FT RL problem, an SPE describes a scenario in which the controller optimally responds against stoppages at the set of states that inflict the greatest costs to the controller. In particular, we will demonstrate that $\hat{\pi}\in\Pi$ is a BR to a system that undergoes adversarial attacks.
\begin{proposition}\label{saddle_pt_theorem}
The pair $(\hat{\tau},\hat{\pi})\in\mathcal{V}   \times\Pi$ consists of BR strategies and constitutes an SPE.
\end{proposition}
\begin{proof}
The proposition follows from the fact that if either player plays a Markov strategy then their opponent's best-response is a Markov strategy. Moreover, $\hat{\tau}$ is a BR strategy for player 2 (recall Definition 3). Moreover, by Theorem 1 (commuting the $\max$ and $\min$ operators) we observe that $\hat{\pi}$ is a BR strategy for player 1. 
\end{proof}

By Prop. \ref{saddle_pt_theorem}, when the pair $(\hat{\tau},\hat{\pi})$ is played, each player executes its BR strategy. The strategic response then induces FT behaviour by the controller.
We now turn to the existence and characterising the optimal stopping time for player 2. The following result establishes its existence. 
\begin{theorem}\label{stopping_time_theorem}
There exists an $\mathcal{F}$-measurable stopping time:\\
$
    \hat{\tau}=\min\left\{k\in\mathcal{T}   \Big|G(s_k)\leq \underset{v\in\mathcal{V}   }{\min}\hspace{0.2 mm}\underset{\pi\in\Pi}{\max} J^{v,\pi}[s_k]\right\},\;\; a.s.  
$\end{theorem}
The theorem characterises and establishes the existence of the player 2 optimal stopping time which, when executed by the adversary, induces an FT control by the controller. 
\begin{proof}[Proof of Theorem \ref{stopping_time_theorem}]
For any $m\in\mathbb{N}$ we have that:
\begin{align}
&\underset{\pi\in\Pi}{\max}\;J^{{\tau},{\pi}}[s]
\begin{aligned}
\geq \;&\underset{\pi\in\Pi}{\max}\;J^{\tau\wedge m,{\pi}}[s]
-\sum_{t=m}^\infty\gamma^t\underset{\pi\in\Pi}{\max}\;\left(|R(s_t,a_t)|+|G(s_t)|\right).\label{theorem_2_start}
\end{aligned}
\end{align}
We now apply the $\min$ operator to both sides of (\ref{theorem_2_start}) which gives:
\begin{align}\nonumber
\underset{\tau\in\mathcal{T}}{\min}\;\underset{\pi\in\Pi}{\max}\;J^{{\tau},{\pi}}[s]
\geq \underset{\tau\in\mathcal{T}}{\min}\;\underset{\pi\in\Pi}{\max}\;J^{\tau\wedge m,{\pi}}[s]
-\sum_{t=m}^\infty\gamma^t\underset{\pi\in\Pi}{\max}\;\left(|R(s_t,a_t)|+|G(s_t)|\right).
\end{align}
After taking expectations, we find that:
\begin{align}
&\mathbb{E}\left[\underset{\tau\in\mathcal{T}}{\min}\;\underset{\pi\in\Pi}{\max}\;J^{{\tau},{\pi}}[s]\right]
\\&\begin{aligned}
\geq \;&\mathbb{E}\left[\underset{\tau\in\mathcal{T}}{\min}\;\underset{\pi\in\Pi}{\max}\;J^{\tau\wedge m,{\pi}}[s]\right]
-\sum_{t=m}^\infty\gamma^t\mathbb{E}\left[\underset{\pi\in\Pi}{\max}\;\left(|R(s_t,a_t)|+|G(s_t)|\right)\right].\label{theorem_1_after_expectations}
\end{aligned}
\end{align}
Now by Jensen's inequality and, using the stationarity of the state process (recall the expectation is taken under $\pi$) we have that:
\begin{align}\nonumber
&\mathbb{E}\left[\underset{\pi\in\Pi}{\max}\;\left(|R(s_t,a_t)|+|G(s_t)|\right)\right]
\\&\geq \underset{\pi\in\Pi}{\max}\;\mathbb{E}\left[\left(|R(s_t,a_t)|+|G(s_t)|\right)\right]
    =\mathbb{E}[|R(s_0,\cdot)|]+\mathbb{E}[|G(s_0)|].\label{jensens_1_theorem_1}
\end{align}
By standard arguments of dynamic programming, the value of the game with horizon $n$ can be obtained from $n$ iterations of the dynamic recursion; in particular, we have that:
\begin{align}
\underset{\tau\in\mathcal{T}}{\min}\;\underset{\pi\in\Pi}{\max}\;J^{\tau\wedge m,{\pi}}[s]=T^m G(s).   \label{dynamic_prog_recursion} 
\end{align}
Inserting (\ref{jensens_1_theorem_1}) and (\ref{dynamic_prog_recursion}) into (\ref{theorem_1_after_expectations}) gives:
\begin{align}
&\mathbb{E}\left[\underset{\tau\in\mathcal{T}}{\min}\;\underset{\pi\in\Pi}{\max}\;J^{{\tau},{\pi}}[s]\right]
\geq\mathbb{E}\left[T^m G(s)\right]
-c(m)\nonumber
\\&= \lim_{m\to\infty}\left[\mathbb{E}\left[T^m G(s)\right]-c(m)\right]
= \mathbb{E}\left[J^{\hat{\tau},\hat{\pi}}[s]\right],\label{theorem_2_proof_s1}
\end{align}
where $c(m):=\frac{\gamma^{m}}{1-\gamma}(\mathbb{E}[|R(s_0,\cdot)|]+\mathbb{E}[|G(s_0)|])$ so that $\underset{m\to\infty}{\lim}c(m)=0$.
Hence, we find that:
\begin{align}
\mathbb{E}\left[J^{\hat{\tau},\hat{\pi}}[s]\right]\leq \mathbb{E}\left[\underset{\tau\in\mathcal{T}}{\min}\;\underset{\pi\in\Pi}{\max}\;J^{{\tau},{\pi}}[s]\right],
\end{align}
 we deduce the result after noting that $G(s_{\tau})=J^{\tau,\cdot}[s_{\tau}]$ by definition of $G$.
\end{proof}
Having shown the existence of the optimal stopping time $\tau^\star$, by Theorem \ref{stopping_time_theorem} and Theorem \ref{existence_theorem}, we find:
\begin{theorem}\label{optim_stop_SG_time_equiv_prop}
Let $\hat{\tau}$ be the player 2 optimal stopping time defined in (\ref{stopping_time_theorem}) and let $\tau^\star$ be the optimal stopping time for the robust OSP (c.f. (\ref{robust_optimal_stopping_time})) then $
    \tau^\star=\hat{\tau}.
$
\end{theorem}
Theorem \ref{optim_stop_SG_time_equiv_prop} establishes an equivalence between the robust OSP and the SG of control and stopping hence, any method that computes $\hat{\tau}$ for the SG yields a solution to the robust OSP.   

\section{Simulation-Based Value Iteration}\label{approx_scheme}
We now develop a \textit{simulation-based} value-iterative scheme. We show that the method produces an iterative sequence that converges to the value of the game from which the optimal controls can be extracted. The method is suitable for environments in which the transition model and reward functions are not known to either player.  We defer the proofs of this section to the appendix.

The fixed point property of the game established in Theorem \ref{fixed_point_theorem} immediately suggests a solution method for finding the value. In particular, we may seek to solve the fixed point equation (FPE) $
    J^\star=TJ^\star$.
Direct approaches at solving the FPE are not generally fruitful as closed solutions are typically unavailable. To compute the value function, we develop an iterative method that tunes weights of a set of basis functions $\{\phi_k:\mathbb{R}^p\to\mathbb{R}|k\in 1,2,\ldots D\}$  to approximate $J^\star$ through
simulated system trajectories and associated costs. Algorithms of this type were first introduced by Watkins \citep{watkins1992q} as an approximate dynamic programming method and have since been augmented to cover various settings. Therefore the following can be considered as a generalised Q-learning algorithm for zero-sum controller stopper games. 

Let us denote by $\Phi r:=\sum_{j=1}^Dr(j)\phi_j$ an operator representation of the basis expansion. The algorithm
is initialised with weight vector $r_0=(r_0(1),\ldots, r_0(P))'\in\mathbb{R}^d$. Then as the trajectory $\{s_t|t=0,1,2,\ldots\}$ is simulated, the algorithm produces an updated series of vectors $\{r_t|t=0,1,2,\ldots\}$ by the update:
\begin{align}
&r_{t+1}\nonumber    =r_t+\gamma\phi(s_t)\Big(\underset{a\in\mathcal{A}}{\max}\;R^a_{s_t}+\gamma\min\left\{(\phi r_t)(s_{t+1}),G(s_{t+1})\right\}-(\phi r_t)(s_t)\Big).\label{approx_dynm_r_update}
\end{align}

Theorem \ref{theorem_approx_convergence} demonstrates that the method converges to an approximation of $J^\star$. We provide a bound for the approximation error in terms of the basis choice.

We define the function $Q^\star$ which the algorithm approximates by:
\begin{align}
    Q^\star(s)=\underset{a\in\mathcal{A}}{\max}\;R^a_s+\gamma PJ^\star[s],\qquad \forall s\in\mathcal{S}
\end{align}

We later show that $Q^\star$ serves to approximate the value $J^\star$. In particular, we show that the algorithm generates a sequence of weights $r_n$ that converge to a vector $r^\star$ and that $\Phi r^\star$, in turn approximates $Q^\star$. To complete the connection, we provide a bound between the outcome of the game when the players use controls generated by the algorithm.  

We introduce our player 2 stopping criterion which now takes the form: 
\begin{align}
\hat{\tau}=\min\{t|G(s_t)\leq Q^\star(s_t)\} \label{approx_stopping_rule}.
\end{align}
Let us define a orthogonal projection $\Pi$ and the function $F$ by the following:
\begin{align}
\Pi Q :=\underset{\bar{Q}\in\{\Phi r|r \in\mathbb{R}^p\}}{\arg\min}\|\bar{Q}-Q\|,   FQ:=\underset{a\in\mathcal{A}}{\max}\; R^a_s+\gamma P\min\{G,Q\}.    
\end{align}
We now state the main results of the section:
\begin{theorem}\label{theorem_approx_convergence} $r_n$ converges to $r^\star$ where $r^\star$ is the unique solution: $
    \Pi F(\Phi r^\star)=\Phi r^\star.$
\end{theorem}
The following results provide approximation bounds when employing the projection $\Pi$:
\begin{theorem}\label{theorem_approx_prog_bounds} Let $\hat{\tau}=\min\left\{k\in\mathcal{V}   \Big|G(s_k)\leq (\Phi r^\star)(s_k)\right\}$, then the following hold:
\begin{align}
\left\|\Phi r^\star-Q^\star\right\|&\leq \left(\sqrt{1-\gamma^2}\right)^{-1}\left\|\Pi Q^\star-Q^\star\right\|,
\\
\mathbb{E}\left[J^\star-J^{\tilde{\tau},\tilde{\pi}}\right]&\leq {2}{\left[(1-\gamma)\sqrt{1-\gamma^2})\right]^{-1}}\left\|\Pi Q^\star-Q^\star\right\|.
\end{align}
\end{theorem}
Hence the error bound in approximation of $J^\star$ is determined by the goodness of the projection.

Theorem \ref{theorem_approx_convergence} and Theorem \ref{theorem_approx_prog_bounds} thus enable the FT RL problem to be solved by way of simulating the behaviour of the environment and using the update rule (\ref{approx_dynm_r_update}) to approximate the value function. Applying the stopping rule in (\ref{approx_stopping_rule}), by Theorem \ref{theorem_approx_prog_bounds} and Theorem \ref{fixed_point_theorem}, means the pair $(\tilde{\tau},\tilde{\pi})$ is generated where the policy $\tilde{\pi}$ approximates the policy $\hat{\pi}$ which is FT against adversarial stoppages and faults. 
\section*{Conclusion}
In this paper, we tackled the problem of fault-tolerance within RL in which the controller seeks to obtain a control that is robust against catastrophic failures. To formally characterise the optimal behaviour, we constructed a new discrete-time SG of control and stopping. We established the existence of an equilibrium value then, using a contraction mapping argument, showed that the game can be solved by iterative application of a Bellman operator and constructed an approximate dynamic programming algorithm so that the game can be solved by simulation.
%
%
%
 \bibliographystyle{{SpringSTYLE15pg/splncs04.bst}}

\clearpage{\Large\section*{Appendix}}
\subsection*{Assumptions}
Our results are built under the following assumptions:\\

\noindent\underline{Assumption A.1.} Stationarity: the expectations $\mathbb{E}$ are taken w.r.t. a stationary distribution so that for any measurable function $f$ we have $\mathbb{E}\left[f(s)\right]=\mathbb{E}\left[f(s_k)\right]$ for any $k\geq 0$ where $s:=s_0$. \\

\noindent\underline{Assumption A.2.} Ergodicity:  i) Any invariant random variable of the state process is $P-$almost surely ($P-$a.s.) a constant.\\

\noindent\underline{Assumption A.3.} Markovian transition dynamics: the transition probability function $P$ satisfies the following equality: $P(s_{k+1}\in A|\mathcal{F}_k)=P(s_{k+1},A)$ for any $A\in\mathcal{B}(\mathbb{R}^p)$.\\

\noindent\underline{Assumption A.4.} The constituent functions $\{R,G\}$ in $J$ are square integrable: that is, $R,G\in \mathbb{L}_2(\mu)$.
\subsection*{Additional Lemmata}
We begin the analysis with some preliminary lemmata and definitions which are useful for proving the main results.

\begin{customdefinition}{A.1}
An operator $T:\mathcal{V}\to\mathcal{V}$ is said to be a \textbf{contraction} w.r.t a norm $\|\cdot\|$ if there exists a constant $c\in[0,1[$ s.th for any $V_1,V_2\in \mathcal{V}$ we have that:
\begin{align}
    \|TV_1-TV_2\|\leq c\|V_1-V_2\|.
\end{align}
\end{customdefinition}

\begin{customdefinition}{A.2}
An operator $T:\mathcal{V}\to \mathcal{V}$ is \textbf{non-expansive} if $\forall V_1,V_2\in \mathcal{V}$ we have:
\begin{align}
    \|TV_1-TV_2\|\leq \|V_1-V_2\|.
\end{align}
\end{customdefinition}
\begin{customdefinition}{A.3}
The \textbf{residual} of a vector $V\in\mathcal{V}$ w.r.t the operator $T:\mathcal{V}\to \mathcal{V}$ is:
\begin{align}
    \epsilon_T(V):= \|TV-V\|.
\end{align}
\end{customdefinition}
\begin{customlemma}{A.1}\label{max_l.val_inequality}
Define ${\rm val}^+[f]:=\min_{b\in\mathbb{B}}\max_{a\in\mathbb{A}}f(a,b)$ and define\\ ${\rm val}^-[f]:=\max_{a\in\mathbb{A}}\min_{b\in\mathbb{B}}f(a,b)$, then for any $b\in\mathbb{B}$ we have that for any $f,g\in \mathbb{L}$ and for any $c\in\mathbb{R}_{>0}$:
\begin{align}\nonumber
\left|\max_{a\in\mathbb{A}}f(a,b)-\max_{a\in\mathbb{A}}g(a,b)\right|&\leq c
\implies
\left|{\rm val}^-[f]-{\rm val}^-[g]\right|\leq c.
\end{align}
\end{customlemma}
\begin{customlemma}{A.2}\label{max_min_inequality_2} For any $f,g,h\in \mathbb{L}$ and for any $c\in\mathbb{R}_{>0}$ we have that:
\begin{align}\nonumber
 \left\| f- g\right\|\leq c
\implies
\left\|\min\{ f,h\}-\min\{ g,h\}\right\|\leq c.
\end{align}
\end{customlemma}
\begin{customlemma}{A.3}\label{max_triple_inequality} Let the functions $f,g,h\in \mathbb{L}$ then
\begin{align}
\left\|\max \{ f,h\}-\max\{g,h\}\right\|&\leq \|f-g\|.
\end{align}
\end{customlemma}
The following lemma, whose proof is deferred is a required result for proving the contraction mapping property of the operator $T$.
\begin{customlemma}{A.4}\label{non_expansive_P}
The probability transition kernel $P$ is non-expansive, that is:
\begin{align}
    \|PV_1-PV_2\|\leq \|V_1-V_2\|.
\end{align}
\end{customlemma} 
The following estimates provide bounds on the value $J^\star$ which we use later in the development of the iterative algorithm. We defer the proof of the results to the appendix.
\begin{customlemma}{A.5}\label{residual_lemma}
Let $T:\mathcal{V}\to\mathcal{V}$ be a contraction mapping in $\|\cdot\|$ and let $J^\star$ be a fixed point so that $TJ^\star=J^\star$ then there exists a constant $c\in [0,1[$ s.th:
\begin{align}
    \|J^\star-J\|\leq (1-c)^{-1}\epsilon_{T}(J).
\end{align}
\end{customlemma}
\begin{customlemma}{A.6}\label{residual_2}
Let $T_1:\mathcal{V}\to\mathcal{V},T_2:\mathcal{V}\to\mathcal{V}$ be contraction mappings and suppose there exists vectors $J^\star_1,J^\star_2$ s.th  $T_1J^\star_1=J^\star_1$ and $T_2J^\star_2=J^\star_2$ (i.e. $J^\star_1,J^\star_2$ are fixed points w.r.t $T_1$ and $T_2$ respectively) then  $\exists c_1,c_2\in [0,1[$ s.th:
\begin{align*}
    \left\|J^\star_1-J^\star_2\right\|\leq \left(1-\{c_1\wedge c_2\}\right)^{-1}\left(\epsilon_{T_1}(J)-\epsilon_{T_2}(J)\right).
\end{align*}
\end{customlemma}
\begin{customlemma}{A.7}\label{T_operator_properties}
The operator $T$ satisfies the following:
\begin{enumerate}
\item(Monotonicity) For any $J_1,J_2\in \mathbb{L}_2$ s.th. $J_1(s)\leq J_2(s)$ then $TJ_1\leq TJ_2$.
\item(Constant shift) Let $I(s)\equiv \textbf{1}$ be the unit function, then for any $J\in \mathbb{L}_2$ and for any scalar $\alpha\in \mathbb{R}$, $T$ satisfies $T(J+\alpha I)(s)=TJ(s)+\alpha I(s)$.
\end{enumerate}
\end{customlemma}
\subsection*{Proof of Results}
\begin{proof}[Proof of Lemma \ref{max_l.val_inequality}]
We begin by noting the following inequality for any
$f:\mathcal{V}\times\mathcal{V}\to\mathbb{R},g:\mathcal{V}\times\mathcal{V}\to\mathbb{R}$ s.th. $f,g\in \mathbb{L}$ we have that for all $b\in\mathcal{V}$:
\begin{align}
\left|\underset{a\in\mathcal{V}}{\max}\:f(a,b)-\underset{a\in\mathcal{V}}{\max}\: g(a,b)\right| \leq \underset{a\in\mathcal{V}}{\max}\: \left|f(a,b)-g(a,b)\right|.    \label{lemma_1_basic_max_ineq}
\end{align}
From (\ref{lemma_1_basic_max_ineq}) we can straightforwardly derive the fact that for any $b\in\mathcal{V}$:
\begin{align}
\left|\underset{a\in\mathcal{V}}{\min}\: f(a,b)-\underset{a\in\mathcal{V}}{\min}\: g(a,b)\right| \leq \underset{a\in\mathcal{V}}{\max}\: \left|f(a,b)-g(a,b)\right|,   \label{lemma_1_max_ineq_min_version} 
\end{align}
(this can be seen by negating each of the functions in (\ref{lemma_1_basic_max_ineq}) and using the properties of the $\max$ operator).

Assume that for any $b\in\mathcal{V}$ the following inequality holds: 
\begin{align}
\underset{a\in\mathcal{V}}{\max}\: \left|f(a,b)-g(a,b)\right|  \leq c \label{lemma_1_assumption} 
\end{align}
Since (\ref{lemma_1_max_ineq_min_version}) holds for any $b\in\mathcal{V}$ and, by (\ref{lemma_1_basic_max_ineq}), we have in particular that
\begin{align}
&\nonumber\left|\underset{b\in\mathcal{V}}{\max}\;\underset{a\in\mathcal{V}}{\min}\: f(a,b)-\underset{b\in\mathcal{V}}{\max}\;\underset{a\in\mathcal{V}}{\min}\: g(a,b)\right| 
\\&\nonumber\leq
\underset{b\in\mathcal{V}}{\max}\left|\underset{a\in\mathcal{V}}{\min}\: f(a,b)-\underset{a\in\mathcal{V}}{\min}\: g(a,b)\right| 
\\&\leq \underset{b\in\mathcal{V}}{\max}\;\underset{a\in\mathcal{V}}{\max}\: \left|f(a,b)-g(a,b)\right|  \leq c,\label{lemma_1_max_ineq_fin} 
\end{align}
whenever (\ref{lemma_1_assumption}) holds which gives the required result.
\end{proof}
Lemma \ref{max_min_inequality_2} and Lemma \ref{max_triple_inequality} are given without proof but can be straightforwardly checked.
\begin{proof}[Proof of Lemma \ref{non_expansive_P}]
The proof is standard, we give the details for the sake of completion. Indeed, using the Tonelli-Fubini theorem and the iterated law of expectations, we have that:
\begin{align*}
&\|PJ\|^2=\mathbb{E}\left[(PJ)^2[s_0]\right]\\&=\mathbb{E}\left(\left[\mathbb{E}\left[J[s_1]|s_0\right]\right)^2\right]
\leq \mathbb{E}\left[\mathbb{E}\left[J^2[s_1]|s_0\right]\right] 
= \mathbb{E}\left[J^2[s_1]\right]=\|J\|^2,
\end{align*}
where we have used Jensen's inequality to generate the inequality. This completes the proof.
\end{proof}

\begin{proof}[Proof of Lemma \ref{residual_lemma}]
The proof follows almost immediately from the triangle inequality, indeed for any $J\in \mathbb{L}_2$:
\begin{align}
\|J^\star-J\|=\|TJ^\star-J\|
\leq \gamma\|J^\star-J\|+\|TJ-J\|,
\end{align}
where we have added and subtracted $TJ$ to produce the inequality. The result then follows after inserting the definition of $\epsilon_T(J)$. 
\end{proof}
\begin{proof}[Proof of Lemma \ref{residual_2}]
The proof follows directly from Lemma \ref{residual_lemma}.
Indeed, we observe that for any $J\in \mathbb{L}_2$ we have
\begin{align}
\|J^\star_1-J^\star_2\|
\leq \|J^\star_1-J\|+\|J^\star_2-J\|,
\end{align}
where we have added and subtracted $J$ to produce the inequality. The result then follows from Lemma \ref{residual_lemma}.
\end{proof}
\begin{proof}[Proof of Lemma \ref{T_operator_properties}]
Part 2 immediately follows from the properties of the $\max$ and $\min$ operators. 
It remains only to prove part 1.\\
We seek to prove that for any $s\in\mathcal{S}$, if $J\leq \bar{J}$ then
\begin{align}
\begin{aligned}\min_{\tau\in\mathcal{T}}&\left\{\max_{a\in A} R^a_s+\gamma\sum_{s'\in\mathcal{S}}P^a_{ss'}J^{\tau,\pi}[s'],G(S_\tau)\right\}
\\&\qquad\qquad\qquad-\min_{\tau\in\mathcal{T}}\left\{\max_{a\in A} R^a_s+\gamma\sum_{s'\in\mathcal{S}}P^a_{ss'}\bar{J}^\pi[s'],G(S_\tau)\right\}\leq 0
\end{aligned}
\end{align}
We begin by firstly making the following observations:\\
1. For any $x,y,h\in \mathcal{V}$ 
\begin{align}
x\leq y \implies \min\{x,h\}\leq \min\{y,h\}. \label{basic_min_result}
\end{align}
2. For any $f,g,h\in \mathbb{L}_2$
\begin{align}
\left|\underset{x\in\mathcal{V}}{\max}\;f(x)-\underset{x\in\mathcal{V}}{\max}\;g(x)\right|\leq
\underset{x\in\mathcal{V}}{\max}\left|f(x)-g(x)\right|.
\label{basic_max_result}
\end{align}
Assume that $J\leq \bar{J}$, then we observe that:
\begin{align}
&\begin{aligned}\underset{a\in\mathcal{A}}{\max}&\left\{R^a_s+\gamma\sum_{s'\in\mathcal{S}}P^a_{ss'}J^{\tau,\pi}[s']\right\}
-\underset{a\in\mathcal{A}}{\max}\left\{R^a_s+\gamma\sum_{s'\in\mathcal{S}}P^a_{ss'}\bar{J}^\pi[s']\right\}
\end{aligned}
\\&\nonumber\leq \gamma\max_{a\in A}\left\{\sum_{s'\in\mathcal{S}}P^a_{ss'}\left(J^{\tau,\pi}[s']-\bar{J}^\pi[s']\right)\right\}
\\&=\gamma\left(\left(PJ\right)-\left(P\bar{J}\right)\right)\nonumber
\leq J-\bar{J}\leq 0,
\end{align}
where we have used (\ref{basic_max_result}) in the penultimate line. The result immediately follows after applying (\ref{basic_min_result}).
\end{proof}


The proofs of the results in Sec. \ref{approx_scheme} are constructed in a similar fashion that in (Bertsekas, 2008) (approximate dynamic programming). However, the analysis incorporates some important departures due to the need to accommodate the actions of two players that operate antagonistically.

We now prove the first of the two results of Sec. \ref{approx_scheme}.
\begin{proof}[Proof of Theorem \ref{theorem_approx_convergence}]
We firstly notice the construction of $\hat{\tau}$ given by
\begin{align}
\hat{\tau}=\min\{t|G(s_t)\leq Q^\star\},
\end{align}
is sensible since we observe that 
\begin{align}\nonumber
&\min\{t|G(s_t)\leq J^\star\}
\\\nonumber=&\min\{t|G(s_t)\leq \min\{G(s_t), Q^\star(s_t)\}
\\\nonumber=&\min\{t|G(s_t)\leq Q^\star\}.
\end{align}
\noindent\underline{Result 1}
\\\noindent\underline{Step 1}
Our first step is to prove the following bound:
\begin{align}
\left\|FQ-F\bar{Q}\right\|\leq\gamma \left\| Q- \bar{Q}\right\|.
\end{align}
\begin{proof}
\begin{align}\nonumber
&\left\|\underset{a\in\mathcal{A}}{\max}\;R^a_s+\gamma P\min\{G,Q\}-\left(\underset{a\in\mathcal{A}}{\max}\;R^a_s+\gamma P\min\{G,\bar{Q}\}\right)\right\|
\\&\nonumber=\gamma\left\| P\min\{G,Q\}- P\min\{G,\bar{Q}\}\right\|
\\&\nonumber\leq\gamma \left\| \min\left\{G,Q\right\}- \min\left\{G,\bar{Q}\right\}\right\|
\\&\nonumber\leq\gamma \left\| Q- \bar{Q}\right\|.
\end{align}
which is the required result.
\end{proof}
\noindent\underline{Step 2}\\
Our next task is to prove that the quantity $Q^\star$ is a fixed point of $F$ and hence we can apply the operator $F$ to achieve the approximation of the value.
\begin{proof}
Using the definition of $T$ (c.f. (13) we find that:
\begin{align*}
&\hspace{-15 mm}J^\star=TJ^\star
\iff
\underset{a\in\mathcal{A}}{\max}\;R^a_s+\gamma PJ^\star
\\&\hspace{-16 mm}=\underset{a\in\mathcal{A}}{\max}\;R^a_s+\gamma P\min\left\{\underset{a\in\mathcal{A}}{\max}\;R^a_s+\gamma PJ,G\right\}
\\&\hspace{-18 mm}
\iff
\\&\hspace{-16 mm}Q^\star=\underset{a\in\mathcal{A}}{\max}\;R^a_s+\gamma P\min\left\{Q^\star,G\right\}
\\&\hspace{-18 mm}
\iff
\\&Q^\star= F Q^\star.
\end{align*}
\end{proof}
\noindent\underline{Step 3}\\
We now prove that the operator $\Pi F$ is a contraction on $Q$, that is the following inequality holds:
\begin{align*}
\left\|\Pi F Q -\Pi F\bar{Q}\right\|\leq \gamma \left\|Q-\bar{Q}\right\|.    
\end{align*}
\begin{proof}
The proof follows straightforwardly by the properties of a projection mapping:  
\begin{align*}
\left\|\Pi F Q - \Pi F\bar{Q}\right\|\leq  \left\|FQ-F\bar{Q}\right\|\leq  \gamma\left\|Q-\bar{Q}\right\|.   \end{align*}
\end{proof}
\noindent\underline{Step 4}\\
\begin{align}
\left\|\Phi r^\star- Q^\star\right\|
\leq\frac{1}{\sqrt{1-\gamma^2}}\left\|\Pi Q^\star- Q^\star\right\|.\label{step_4_result}
\end{align}
The result is proven using the orthogonality of the (orthogonal) projection and by the Pythagorean theorem. Indeed, we have that: 
\begin{proof}
\begin{align*}
&\left\|\Phi r^\star- Q^\star\right\|^2=  
\left\|\Phi r^\star- \Pi Q^\star\|^2+\|\Pi Q^\star- Q^\star\right\|^2
\\&=  
\left\|\Pi F\Phi r^\star- \Pi Q^\star\right\|^2+\left\|\Pi Q^\star- Q^\star\right\|^2
\\&=  
\left\|\Pi F\Phi r^\star- \Pi Q^\star\right\|^2+\left\|\Pi Q^\star- Q^\star\right\|^2
\\&\leq\gamma^2\left\|\Phi r^\star-  Q^\star\right\|^2+\left\|\Pi Q^\star- Q^\star\right\|^2.
\end{align*}
Hence, we find that
\begin{align*}
\left\|\Phi r^\star- Q^\star\right\|
\leq\frac{1}{\sqrt{1-\gamma^2}}\left\|\Pi Q^\star- Q^\star\right\|,
\end{align*}
which is the required result.
\end{proof}
\noindent\underline{Result 2}\\
\begin{align}
\mathbb{E}\left[J^\star[s]\right]-\mathbb{E}\left[J^{\tilde{\tau},\tilde{\pi}}[s]\right]\leq \frac{2}{[(1-\gamma)\sqrt{1-\gamma^2}]}\|\Pi Q^\star -Q^\star\|.
\end{align}
\begin{proof}
The proof by Jensen's inequality, stationarity and the non-expansive property of $P$. In particular, we have
\begin{align}\nonumber
&\mathbb{E}\left[J^\star[s]\right]-\mathbb{E}\left[J^{\tilde{\tau},\tilde{\pi}}[s]\right]
\\&\nonumber=\mathbb{E}\left[PJ^\star[s]\right]-\mathbb{E}\left[PJ^{\tilde{\tau},\tilde{\pi}}[s]\right]
\\&\nonumber\leq\left|\mathbb{E}\left[PJ^\star[s]\right]-\mathbb{E}\left[PJ^{\tilde{\tau},\tilde{\pi}}[s]\right]\right|
\\&\leq\|PJ-PJ^{\tilde{\tau},\tilde{\pi}}\|.\label{J_bound_supp_end}
\end{align}
Inserting the definitions of $Q^\star$ and $\tilde{Q}$ into 
(\ref{J_bound_supp_end}) then gives:
\begin{align}
    \mathbb{E}\left[J^\star[s]\right]-\mathbb{E}\left[J^{\tilde{\tau},\tilde{\pi}}[s]\right]\leq \frac{1}{\gamma}\|Q^\star-\tilde{Q}\|
.\end{align}
It remains therefore to place a bound on the term $\|Q^\star-\tilde{Q}$. We observe that by the triangle inequality and the fixed point properties of $F$ on $Q$ and $\tilde{F}$ on $\tilde{Q}$ we have
\begin{align}
\|Q^\star-\tilde{Q}\|&\leq \|Q^\star-F(\Phi r^\star)\|+\|\tilde{Q}-F(\Phi r^\star)\|   
\\&\leq\gamma\left\{ \|Q^\star-\Phi r^\star\|+\|\tilde{Q}-\Phi r^\star\|\right\}
\\&\leq\gamma\left\{ 2\|Q^\star-\Phi r^\star\|+\|Q^\star-\tilde{Q}\|\right\}.
\end{align}
So that
\begin{align}
\|Q^\star-\tilde{Q}\|&\leq \frac{2\gamma}{1-\gamma}\|Q^\star-\Phi r^\star\|.
\end{align}
The result then follows after substituting the result of step 4 (\ref{step_4_result}).
\end{proof}
Let us now define the following quantity:
\begin{align}
HQ(s):=\begin{cases}
G(s)&\mbox{if } G(s)\leq (\Phi r^\star)(s)
\\Q(s)& {\rm otherwise},
\end{cases}
\end{align}
and
\begin{align}
\tilde{F}Q:=\underset{a\in\mathcal{A}}{\max}\;R^a_s+\gamma PHQ.    
\end{align}
\noindent\underline{Step 5}
\begin{align}
\left\|\tilde{F}Q-\tilde{F}\bar{Q}\right\|\leq \gamma \left\|Q-\bar{Q}\right\|  
\end{align}
\begin{proof}
\begin{align*}
&\left\|\tilde{F}Q-\tilde{F}\bar{Q}\right\|=\left\|\underset{a\in\mathcal{A}}{\max}\;R^a_s+\gamma PHQ-\left(\underset{a\in\mathcal{A}}{\max}\;R^a_s+\gamma PH\bar{Q}\right)\right\|
\\&=\gamma\left\| PHQ- PH\bar{Q}\right\|
\\&\leq\gamma\left\|HQ- H\bar{Q}\right\|
\\&=\gamma\left\|\min\{G,Q\}- \min\{G,\bar{Q}\}\right\|
\\&\leq\gamma\left\|Q- \bar{Q}\right\|.
\end{align*}
We now prove that $    \tilde{Q}=\underset{a\in\mathcal{A}}{\max}\;R^a_s+\gamma PJ^{\pi,\tilde{\tau}}$ is a fixed point.
\begin{align*}
&H\tilde{Q}=H\left(\underset{a\in\mathcal{A}}{\max}R^a_s+\gamma PJ^{\pi,\tilde{\tau}}\right)
\\&=\begin{cases}
G(s)&\mbox{if } G(s)\leq (\Phi r^\star)(s)
\\\underset{a\in\mathcal{A}}{\max}R^a_s+\gamma PJ^{\pi,\tilde{\tau}}& {\rm otherwise}
\end{cases}
\\&=J^{\pi,\tilde{\tau}}
\end{align*}
\end{proof}
Let us now define the following quantity:
\begin{align*}
s(z,r):=\phi(s)\left(\underset{a\in\mathcal{A}}{\max}\;R^a_s+\gamma\min\left\{(\Phi r)(y),G(y)\right\}-(\Phi r)(s)\right).    \end{align*}
Additionally, we define $\bar{s}$ by the following: 
\begin{align*}
\bar{s}(z,r):=\mathbb{E}\left[s(z_0,r)\right].    
\end{align*}
The components of $s(z,r)$ are then given by: 
\begin{align*}
    s_k\equiv \mathbb{E}\left[\phi_k(s_0)\left(\underset{a\in\mathcal{A}}{\max}\;R^a_s+\gamma\min\left\{(\phi r)(s_0),G(s_0)\right\}-(\phi r)(s_0)\right)\right].
\end{align*}
We now observe that $s_k$ can be described in terms of an inner product. Indeed, using the iterated law of expectations we have that
\begin{align*}
s_k&\equiv \mathbb{E}\left[\Phi_k(s_0)\left(\underset{a\in\mathcal{A}}{\max}\;R^a_s+\gamma\min\left\{(\Phi r)(s_0),G(s_0)\right\}-(\Phi r)(s_0)\right)\right]
\\&
= \mathbb{E}\left[\Phi_k(s_0)\left(\underset{a\in\mathcal{A}}{\max}\;R^a_s+\gamma\mathbb{E}\left[\min\left\{(\Phi r)(s_0),G(s_0)\right\}|s_0\right]-(\Phi r)(s_0)\right)\right]
\\&= \mathbb{E}\left[\Phi_k(s_0)\left(\underset{a\in\mathcal{A}}{\max}\;R^a_s+\gamma P\min\left\{(\Phi r)(s_0),G(s_0)\right\}-(\Phi r)(s_0)\right)\right]
\\&= \left\langle\Phi_k,F(\Phi r)-F(\Phi r)\right\rangle.
\end{align*}
\end{proof}
\begin{proof}[Proof of Theorem \ref{theorem_approx_prog_bounds}]
\noindent Step 5 enables us to use classic arguments for approximate dynamic programming. In particular, following step 5, Theorem \ref{theorem_approx_prog_bounds} follows directly from Theorem 2 in (Tsitsiklis \& Van Roy,
1999) with only a minor adjustment in substituting the $\max$ operator with $\min$.
\end{proof}

\end{document}